\documentclass[aps,prl,preprint,superscriptaddress,floatfix,12pt]{revtex4-1}
\usepackage{amsmath}
\usepackage{amsthm}
\usepackage{amssymb}
\usepackage{amsbsy}
\usepackage{amsfonts}
\usepackage{graphicx}
\usepackage{dcolumn}
\usepackage{bm}% bold math
\usepackage{braket}
\usepackage{multirow}

\newtheorem{lemma}{Lemma}
\begin{document}

%\preprint{}

\title{Applicability of Kerker preconditioning scheme to the self-consistent density functional theory calculations of inhomogeneous systems}


\author{Yuzhi Zhou}
\affiliation
{Laboratory of Computational Physics, Huayuan Road 6, Beijing 100088, People's
Republic of China}
\affiliation
{Institute of Applied Physics and Computational Mathematics, Fenghao East Road
2, Beijing 100094, People's Republic of China}
\affiliation
{CAEP Software Center for High Performance Numerical Simulation, Huayuan Road
6, Beijing 100088, People's Republic of China}

\author{Han Wang}
\affiliation
{Institute of Applied Physics and Computational Mathematics, Fenghao East Road
2, Beijing 100094, People's Republic of China}
\affiliation
{CAEP Software Center for High Performance Numerical Simulation, Huayuan Road
6, Beijing 100088, People's Republic of China}

\author{Yu Liu}
\affiliation
{Institute of Applied Physics and Computational Mathematics, Fenghao East Road
2, Beijing 100094, People's Republic of China}
\affiliation
{CAEP Software Center for High Performance Numerical Simulation, Huayuan Road
6, Beijing 100088, People's Republic of China}

\author{Xingyu Gao}
\affiliation
{Laboratory of Computational Physics, Huayuan Road 6, Beijing 100088, People's
Republic of China}
\affiliation
{Institute of Applied Physics and Computational Mathematics, Fenghao East Road
2, Beijing 100094, People's Republic of China}
\affiliation
{CAEP Software Center for High Performance Numerical Simulation, Huayuan Road
6, Beijing 100088, People's Republic of China}

\author{Haifeng Song}
\affiliation
{Institute of Applied Physics and Computational Mathematics, Fenghao East Road
2, Beijing 100094, People's Republic of China}
\affiliation
{CAEP Software Center for High Performance Numerical Simulation, Huayuan Road
6, Beijing 100088, People's Republic of China}

\email[Send correspondence to:]{song_haifeng@iapcm.ac.cn, gao_xingyu@iapcm.ac.cn}

\date{\today}

\begin{abstract}
Kerker preconditioner, based on the dielectric function of homogeneous
electron gas, is designed to accelerate the self-consistent field (SCF)
iteration in the density functional theory (DFT) calculations. However,
question still remains regarding its applicability to the inhomogeneous
systems. In this paper, we develop a modified Kerker preconditioning scheme
which captures the long-range screening behavior of inhomogeneous systems thus
improve the SCF convergence. The effectiveness and efficiency is shown
by the tests on long-\emph{z} slabs of metals, insulators and metal-insulator
contacts. For situations without \emph{a priori} knowledge of the system,
we design the \emph{a posteriori} indicator to monitor if
the preconditioner has suppressed charge sloshing during the iterations. Based on the \emph{a posteriori} indicator, we demonstrate two schemes of the self-adaptive configuration for the SCF iteration.
\end{abstract}


\maketitle
\section{I. Introduction}
Over the past few decades, the Kohn-Sham density functional theory (DFT)
calculation \cite{HohenbergKohn1964,KohnSham1965} has evolved into one of the
most popular \emph{ab initio} approaches for predicting the electronic
structures and related properties of matters. The computational kernel of the
Kohn-Sham DFT calculation is to solve a tangible nonlinear eigenvalue
problem, replacing the original difficult many-body problem
\cite{KohnSham1965}. The Kohn-Sham equation is usually solved by the
self-consistent field (SCF) iteration, which is proved quite reliable and efficient in most cases \cite{KresseFurthmuller1996}. However, the well known "charge sloshing" problem is likely to occur in the SCF iterations as the dimension of the system gets large. The charge sloshing generally refers to the long-wavelength oscillations of the output charge density due to some small changes in the input density during the iterations, and results in a slow convergence or even divergence \cite{Kresse1996CMS,KresseFurthmuller1996,Annett1995}. In some cases it might be referred to the oscillation between different local states of the d or f electrons \cite{MarksLuke2008}. In this work, we concentrate on the former situation.

Given a fixed number of total atoms, charge sloshing and poor SCF
convergence are more prominent and exacerbated in the long-\emph{z} slab
systems, where one dimension of the unit cell is much longer than the other
two's. On the other hand, investigating the properties of the surface and
the interface using DFT calculations has become one important subject in many
scientific and technological fields, such as solid state physics,
semiconductor processing, corrosion, and heterogeneous catalysis
\cite{Hasnip2013,Norskov2011}. The surface/interface is generally simulated by slab model with periodic boundary condition. When one comes to
distinguish the properties between the bulk and the surface, a rather thick
slab is needed to fully restore region with bulk-like properties. One example
is the calculation of the band offsets and valence band alignment at the
semiconductor heterojunctions \cite{Walle1986}. To get a quantitatively
accurate band offset, the lattice of both semiconductors must be extended far
away from the contact region. A similar example is the calculation of work function. Generally speaking, a  thick slab calculation requires relative high accuracy and is very likely to encounter charge sloshing. Effective and efficient mixing schemes are therefore needed to speed up the convergence in the surface/interface calculations.

Practical mixing schemes in the modern DFT code generally takes into account
two aspects: one is combining the results from previous iterations to build
the input for the next step; the other is reflecting the dielectric response of
the system (better known as "preconditioning"). On the first aspect, Pulay
and Broyden-like schemes are well established and widely used
\cite{Pulay1980,Broyden1965}. On the second aspect, Kerker in 1981 proposed
that charge mixing could be preconditioned by a diagonal matrix in the
reciprocal space. This matrix takes the form of inverse dielectric matrix
derived from Thomas-Fermi model of homogeneous electron gas \cite{Kerker1981}.
As pointed out in some literature
\cite{WangLinWang2001,Louie1984,KresseFurthmuller1996}, the preconditioning
matrix should be an approximation to the dielectric function of the system. In
this sense, Kerker preconditioner is ideal for simple metals such as Na
and Al whose valence electrons can be approximated by the homogeneous electron gas. Moreover, for most metallic systems, Kerker preconditioner is a suitable preconditioner since it describes the dielectric responses at the long wavelength limit fairly well.

A natural question is then raised: can Kerker preconditioner be applied to the insulating systems or the inhomogeneous systems such as metal-insulator contact? Efforts have been made to develop effective preconditioning schemes to
accommodate related issues. Kresse \emph{et al.} suggest adding a lower
bound to Kerker preconditioner for the calculation of large insulating systems \cite{KresseFurthmuller1996}. Similarly, Gonze \emph{et al.} realize it with a smoother preconditioning function \cite{abinit2009}. Raczkowski \emph{et al.} solve the Thomas Fermi von Weizs\"{a}cker equation to directly compute the optimized mixing density, in which process the full dielectric function is implicitly solved \cite{WangLinWang2001}. Ho \emph{et al.} \cite{Ho1982}, Sawamura \emph{et al.} \cite{Sawamura2004} and Anglade \emph{et al.} \cite{Anglade2008} adopt preconditioning schemes in which the exact dielectric matrix is computed by the calculated Kohn-Sham orbitals. Shiihara \emph{et al.} recast the Kerker preconditioning scheme in the real space \cite{Shiihara2008}. Lin and Yang further proposed an elliptic preconditioner in the real space method to better accommodates the SCF calculations of large inhomogeneous systems \cite{Lin2013}.

The major part of the above works relies on solving the realistic dielectric response either explicitly or implicitly. However, the extra computational overhead cannot be negligible for large-scale systems.
In practice, the computational expense to achieve the SCF convergence is
more of the concern and a good preconditioner does not necessarily
mean solving the dielectric function as accurately as possible. In this paper,
we focus on extending the applicability of Kerker preconditioning model,
which is based on the simple form of
Thomas-Fermi screening model. This is achieved by modifying Kerker preconditioner
to better capture the long-range screening behavior of the
inhomogeneous systems. For perfect insulating system, we introduce a threshold
parameter to represent the incomplete screening behavior at the long range. With the threshold parameter being set based on the static dielectric constant of the system, the SCF convergence can be reached efficiently and is independent of the system size. For metal-insulator
hybrid systems, the idea of the "effective" conducting electrons is introduced
to approximate the module of the Thomas-Fermi wave vector in the original Kerker
preconditioner. By estimating this module \emph{a priori}, we can
achieve the SCF convergence within 30 iterations in the calculations of
Au-MoS$_2$ slabs with a thickness of 160 {\AA}, saving about 40\% of the SCF iteration steps compared
to the original Kerker scheme. When one does not have sufficient knowledge of the
systems, we design an \emph{a posteriori} indicator to monitor if the charge
sloshing has been suppressed and to guide appropriate parameter setting. Based
on the \emph{a posteriori} indicator, we further present two schemes of self-adaptive
configuration of the SCF iterations. The implementation of our approach requires only small modifications on the original Kerker scheme and the extra computational overhead is negligible.

This paper is organized as follows: In Section II, we will reformulate the Pulay mixing scheme to show the physical meaning of the preconditioner in solving the fixed point equation. In Section III, we will revisit the Thomas-Fermi and Resta screening models to extend the Kerker preconditioner to non-metallic systems. In Section IV, the effectiveness and efficiency of our approach will be examined by numerical examples. Further discussions on this preconditioning technique and the introduction of \emph{a posteriori} indicator and self-adaptive configuration schemes will be given in Section V. Concluding remarks will be presented in the last section.

\section{II. Mathematical framework}
\subsection{A. Simple mixing and preconditioning}
Finding the solution of the Kohn-Sham equation where the output density $n^{\mathrm{out}}({\bf r})$ is equal to the input density $n^{\mathrm{in}}({\bf r})$ can be generalized to the following fixed point equation:
\begin{equation}
    {\bf F}({\bf x}) = {\bf x},
    \label{eq:fixed-point-problem}
\end{equation}
where ${\bf x}$ denotes a vector in many dimensions, e.g. the density is expanded in the dimensions of a set of plane waves. This becomes a minimization problem for the norm of the residual which is defined as
\begin{equation}
    {\bf R}({\bf x}) \equiv {\bf F}({\bf x}) - {\bf x}.
    \label{eq:residual}
\end{equation}
The simplest method for seeking the solution of Eq.~\eqref{eq:fixed-point-problem} is the fixed point iteration:
\begin{equation}
    {\bf x}_{m+1} = {\bf F}({\bf x}_m).
    \label{eq:fixed-point-iteration}
\end{equation}
In the region where ${\bf F}$ is a linear function of ${\bf x}$ and assuming ${\bf x}^*$ is the solution of Eq.~\eqref{eq:fixed-point-problem}, we have
\[
{\bf x}_{m+1}-{\bf x}^* = \left(\frac{\delta {\bf F}}{\delta {\bf x}}\right)^m({\bf x}_1-{\bf x}^*).
\]
Therefore, a necessary condition that guarantees the convergence of the fixed point iteration is
\[
\sigma\left(\frac{\delta {\bf F}}{\delta {\bf x}}\right) < 1,
\]
where $\sigma(A)$ is the spectral radius of the operator or matrix $A$. Unfortunately, in the Kohn-Sham equations, the above condition is generally not satisfied \cite{Annett1995}.

However, the simple mixing can reach convergence as long as
$\displaystyle\sigma\left(\frac{\delta {\bf F}}{\delta {\bf x}}\right)$ is bounded. The simple mixing scheme takes the form:
\begin{equation}
    {\bf x}_{m+1} = {\bf x}_m + P{\bf R}({\bf x}_m),
    \label{eq:simple-mixing}
\end{equation}
where $P$ is the matrix whose size is equal to the number of basis functions. We define the Jacobian matrix:
\begin{equation}
    J \equiv -\frac{\delta {\bf R}}{\delta {\bf x}}
    = I - \frac{\delta {\bf F}}{\delta {\bf x}},
    \label{eq:Jacobian-matrix}
\end{equation}
and denote its value at ${\bf x}^*$ by $J_*$. When ${\bf x}_m$ are sufficiently close to ${\bf x}^*$, the residual propagation of simple mixing Eq.~\eqref{eq:simple-mixing} is given by:
\begin{equation}
    {\bf R}({\bf x}_{m+1}) \approx \left(I-J_*P\right){\bf R}({\bf x}_m).
    \label{eq:error-propagation}
\end{equation}
In some literature \cite{Annett1995,KresseFurthmuller1996,Lin2013}, $P$ is $\alpha I$ with $\alpha$ being a scalar parameter. Then it follows from Eq.~\eqref{eq:error-propagation} that the simple mixing will lead to convergence if:
\begin{equation}
    \sigma\left(I-\alpha J_*\right) < 1.
    \label{eq:conv-cond-1}
\end{equation}
If $\lambda(J_*)$ is an eigenvalue of $J_*$, then the inequality Eq.~\eqref{eq:conv-cond-1} indicates that:
\begin{equation}
    \|1-\alpha\lambda(J_*)\| < 1.
    \label{eq:conv-cond-2}
\end{equation}
Note that $\lambda(J_*)>0$ is referred to as the stability condition of the material in \cite{Bauernschmitt1996JCP}. And it holds in most cases according to the analysis given in \cite{Annett1995}.
Consequently, Eq.~\eqref{eq:conv-cond-2} implies that:
\begin{equation}
    0 < \alpha < \frac{2}{\lambda(J_*)}.
    \label{eq:conv-cond-3}
\end{equation}
When $\lambda(J_*)$ is bounded, it is always possible to find a parameter $\alpha$ to ensure the convergence of the simple mixing scheme. Nevertheless, $\lambda(J_*)$ can become very large in practice, especially in the case of large scale metallic systems, which makes the convergence of the simple mixing extremely slow. Therefore it is desirable to construct effective preconditioning matrix $P$ in Eq.~\eqref{eq:simple-mixing} to speed up the convergence.

Firstly we will show that in the context of the charge mixing, the Jacobian matrix $J$ is just the charge dielectric response function, which describes the charge response to an external charge perturbation. Replacing the ${\bf x}_m$ in Eq.~\eqref{eq:simple-mixing} with charge density ${\bf n}_m$ yields
\begin{equation}
    {\bf n}_{m+1} = {\bf n}_{m} + P \cdot {\bf R}({\bf n}_m).
    \label{eq:charge-mixing}
\end{equation}
For ${\bf R}({\bf n}_m)$, we could expand it near ${\bf n}$ to the linear order
\begin{equation}
{\bf R}({\bf n}) = {\bf R}({\bf n}_m) - J \cdot ({\bf n} - {\bf n}_m),
\label{eq:residual-expand}
\end{equation}
where $J$ in the above equation is just the Jacobian matrix defined earlier in Eq.~\eqref{eq:Jacobian-matrix}. We always want to achieve as much self-consistency as possible in the next step, such that ${\bf R}({\bf n}_{m+1}) \approx 0$. Plugging this into Eq.~\eqref{eq:residual-expand}, we have
\begin{equation}
{\bf n}_{m+1} = {\bf n}_{m} + J^{-1} \cdot {\bf R}({\bf n}_m).
\label{eq:reduced-charge-mixing}
\end{equation}
Comparing Eq.~\eqref{eq:reduced-charge-mixing} with Eq.~\eqref{eq:charge-mixing}, we see that $P = J^{-1}$. The problem then becomes finding a good approximation of the Jacobian matrix $J$. To show that $J$ has the physical meaning of charge dielectric function, we follow Vanderbilt and Louie's procedure in Ref.~\cite{Louie1984}
\begin{equation}
{\bf V}_{m+1} \approx {\bf V}_{m} + U \cdot ({\bf n}_{m+1} - {\bf n}_{m}),
\label{eq:pot-input}
\end{equation}
where the matrix $U$ describes the change in the potential ${\bf V}$ due to a change in the charge density ${\bf n}$. As a result, the output charge density is given by
\begin{equation}
{\bf n}^{\mathrm{out}}_{m+1} \approx {\bf n}^{\mathrm{out}}_{m} + \chi \cdot ({\bf V}_{m+1} - {\bf V}_{m}),
\label{eq:output-pot}
\end{equation}
where $\chi$ is just the electric susceptibility matrix, describing the change in the output charge density due to a change in the potential. Combining Eqs.~\eqref{eq:charge-mixing}, \eqref{eq:residual-expand}, \eqref{eq:pot-input} and \eqref{eq:output-pot} yields
\begin{equation}
J = I - \chi \cdot U.
\label{eq:dielec-general}
\end{equation}
$J$ is often called as the dielectric matrix. According to Vanderbilt and Louie \cite{Louie1984}, $J^{-1}$ is the \emph{charge} dielectric response function which describes the fluctuation in the total charge due to a perturbation from external charge. Adopting the potential mixing, we can also obtain a dielectric response function $(I - U \cdot \chi)^{-1}$ which describes the potential response to an external potential perturbation. Note that the order of the matrix product matters and generally the \emph{charge} dielectric response function and the \emph{potential} dielectric response function are different but closely related.

\subsection{B. Pulay mixing scheme}
Instead of using vector ${\bf x}_m$ only from last step in Eq.~\eqref{eq:simple-mixing}, we can minimize the norm of the residual $\displaystyle ||{\bf R}({\bf x})||$ using the best possible combination of the ${\bf x}_m$ from all previous steps. This is the idea behind the technique called Direct Inversion in the Iterative Subspace (DIIS). It is originally developed by Pulay to accelerate the Hartree-Fock calculation \cite{Pulay1980}. Hence it is often referred to as Pulay mixing in the condensed matter physics community.

An alternative way to derive Pulay method is taking it as the special case of the Broyden's method \cite{Johnson1988PRB}. In the Broyden's second method, a sequence of low-rank modifications are made to modify initial guess of the inverse Jacobian matrix in Eq.~\eqref{eq:Jacobian-matrix} near the solution of Eq.~\eqref{eq:fixed-point-problem}. The recursive formula \cite{MarksLuke2008,Fang2009NLAP} can be derived from the following constrained optimization problem:
\begin{equation}
    \left\{
    \begin{array}{rl}
        \min_H & \frac{1}{2}\|H-H_{m-1}\|_F^2 \\
        \textrm{s.t.} & HY_{m-1} = -S_{m-1},
    \end{array}
    \right.%
    \label{eq:opt-3}
\end{equation}
where $H_{m-1}$ is the approximation to the inverse Jacobian in the $(m-1)$th Broyden update, $S_{m-1}$ and $Y_{m-1}$ are respectively defined as:
\begin{eqnarray}
    S_{m-1} &=& \left(\delta{\bf x}_{m-1}, \cdots, \delta{\bf
    x}_{m-l+1}\right), \nonumber\\
    Y_{m-1} &=& \left(\delta{\bf R}_{m-1}, \cdots, \delta{\bf
    R}_{m-l+1}\right).
    \label{eq:SY}
\end{eqnarray}
It will be later proved in the appendix that the solution to Eq.~\eqref{eq:opt-3} is:
\begin{equation}
    H_{m} = H_{m-1} -
    \left(S_{m-1}+H_{m-1}Y_{m-1}\right)\left(Y_{m-1}^TY_{m-1}\right)^{-1}Y_{m-1}^T.
    \label{eq:Broyden-update}
\end{equation}
We arrive at Pulay mixing scheme by fixing the $H_{m-1}$ in Eq.~\eqref{eq:Broyden-update} to the initial guess $H_1$ of the inverse Jacobian:
\begin{equation}
    H_{m} = H_{1} -
    \left(S_{m-1}+H_{1}Y_{m-1}\right)\left(Y_{m-1}^TY_{m-1}\right)^{-1}Y_{m-1}^T.
    \label{eq:Pulay-update2}
\end{equation}
Then one can follow the quasi Newton approach to generate the next vector:
\begin{eqnarray}
    {\bf x}_{m+1} &=& {\bf x}_m + H_m{\bf R}({\bf x}_m)
    \\\label{eq:quasi-Newton}
    &=& {\bf x}_m + H_1{\bf R}({\bf
    x}_m)-\left(S_{m-1}+H_1Y_{m-1}\right)\left(Y_{m-1}^TY_{m-1}\right)^{-1}Y_{m-1}^T{\bf
    R}({\bf x}_m).
    \label{eq:Broyden-update-matrix}
\end{eqnarray}

We comment that the construction of $H_1$ in Eq.~\eqref{eq:Broyden-update-matrix} is crucial for accelerating the convergence
and is equivalent to the preconditioner for the simple mixing in Eq.~\eqref{eq:simple-mixing}. It is implied by Eq.~\eqref{eq:conv-cond-1} that
preconditioning would be effective if $H_1$ is a good guess of the inverse
dielectric matrix near the solution of Eq.~\eqref{eq:fixed-point-problem}. In
this paper, we concentrate on the Kerker based preconditioning models and appropriate parameterization schemes to capture the long-range dielectric behavior, which turn out to be crucial in improving the SCF convergence.

\section{III. Preconditioning model}
\subsection{A. Thomas-Fermi screening model}
The Thomas-Fermi screening model is the foundation for the Kerker preconditioner. The Thomas-Fermi screening model gives the dielectric response function of the homogeneous electron gas. The dielectric function in the reciprocal space can be expressed as:
\begin{equation}
\varepsilon({\bf q}) = 1 +  \frac{k_{TF}^2}{{\bf q}^2},
\label{eq:dielec-TF2}
\end{equation}
where the Thomas-Fermi vector $k_{TF}$ is given by
\begin{equation}
k_{TF}^2 = 4 \pi e^2 \frac{\delta N}{\delta \mu}.
\end{equation}
The electron number density $N$ is related to the chemical potential $\mu$ through the Fermi-Dirac distribution and dispersion relation of the free electron gas
\begin{equation}
N(\mu) = \int \frac{d {\bf k}}{4 \pi^3} \frac{1}{\exp[\beta(\frac{\hbar^2 {\bf k}^2}{2m_{e}}- \mu)] + 1}.
\label{eq:Fermi-Dirac}
\end{equation}
Now we can derive Kerker preconditioner \cite{Kresse1996CMS,Kerker1981,Lin2013} by inverting the dielectric matrix \footnote{Strictly speaking, the $\varepsilon$ derived here is respect to the potential. However, under the condition of homogeneous system, the potential dielectric response function and charge dielectric response function are same. This is because the $\chi$ in Eq.\eqref{eq:dielec-general} becomes diagonalized thus the product of $\chi \cdot U$ equals to $U \cdot \chi$.}
\begin{equation}
H_{1}^{TF}({\bf q}) = \frac{{\bf q}^2}{{\bf q}^2 + k_{TF}^2}.
\label{eq:original-kerker}
\end{equation}
There are some remarks on the Thomas-Fermi screening model with its implication to Kerker preconditioner and SCF calculations:\\
(i) It can be seen from Eq.~\eqref{eq:dielec-TF2} that the dielectric function diverges quadratically at small $\bf{q}$, which is the mathematical root of the charge sloshing. If a metallic system contains small ${\bf q}$'s, the change in the input charge density will be magnified by the divergence at long wavelength in the dielectric function. This results in large and long-range oscillations in the output charge density, known as the "charge sloshing". Such issue is more prominent in the long-\emph{z} metallic slab systems in which one dimension of the cell is much larger than the rest two. Therefore we use the slab systems for numerical tests.\\
(ii) It is reasonable to ignore the contribution of the exchange-correlation potential in the derivation. In the long wavelength limit, the $1/{\bf q}^2$ divergence at small $\bf q$ is caused by the Coulomb potential while the exchange-correlation potential is local in nature. In this sense, the Thomas-Fermi screening model correctly describe the dielectric behavior of metals at long wavelength, which makes the Kerker preconditioner appropriate for most typical metallic systems.\\
(iii) Even though the dielectric function in Eq.~\eqref{eq:dielec-TF2}
is mounted on the homogeneous electron gas, it still manifests an important feature of the electron screening in the common metallic systems. As mentioned above, $\frac{\delta N}{\delta \mu}$ has the physical meaning of the number of the states in the vicinity of (below and above) the Fermi level. Only these electrons can actively involve in screening since they can adjust themselves to higher unoccupied states to accommodate the change in the potential. Deeper electrons are limited by the high excitation energy due to Pauli exclusion principle. This observation is somehow independent of the band structures of the system. \\
(iv) Following the above point, we further estimate the parameter $k_{TF}$ under the assumption of homogeneous electron gas. Since $\frac{\delta N}{\delta \mu}$ can be approximated by the number of states at the Fermi level, we can write $k_{TF}$ as
\begin{equation}
k_{TF}^2 \approx 4 \pi e^2 N(\varepsilon_{F}) =  \frac{4(3 \pi^2 n_0)^{1/3}}{a_B \pi},
\end{equation}
and
\begin{equation}
a_B = \hbar^2/(me^2) \approx 0.53 \ {\rm \AA},
\end{equation}
where $a_B$ is the Bohr radius and $n_0$ is the total free electron density in the system.
Plugging in numbers, we have the following relation:
\begin{equation}
k_{TF} \approx 2(\frac{n_0}{a_B^3})^{1/6}.
\label{eq:kTF_n}
\end{equation}
In a typical metal, $n_0 \approx 10^{23}$ cm$^{-3}$. Therefore, $k_{TF}
\approx 1 $ \AA$^{-1}$. This is also the default value for Kerker
preconditioner in many simulation packages. As shown later,
Eq.~\eqref{eq:kTF_n} could help us with parameterizing the $k_{TF}$ and facilitate the convergence of the metal-insulator hybrid systems.

\subsection{B. Resta screening model}
The Thomas-Fermi screening model is more appropriate in describing the screening effect in the metallic system. Resta considered the boundary condition of the electrostatic potential for insulators and derived the corresponding screening model \cite{Resta1977}. Rather than the complete screening in the metallic system, the potential is only partially screened beyond some screening length in the insulators. This is characterized by the static dielectric constant $\varepsilon(0)$
\begin{equation}
V({\bf r}) = - \frac{Z}{\varepsilon(0) r}, \quad r \geq R_s,
\end{equation}
where $R_s$ is the screening length and is generally on the order of the lattice constants. According to Resta, the relation between the screening length and the static dielectric constant is given by
\begin{equation}
\varepsilon(0) = \frac{\sinh (q_{0} R_{s})}{q_{0} R_{s}},
\label{eq:Resta-Constr}
\end{equation}
where $q_0$ is a constant related to the valence electron Fermi momentum
$k_{F}$ through
\begin{equation}
q_0 = (4 k_{F}/ \pi)^{1/2}.
\end{equation}
$k_{F}$ is determined by the average valence electron density $n_0$
\begin{equation}
k_{F} = (3 \pi n_0)^{1/3}.
\end{equation}
Under the atomic unit, $q_0$ is in the unit of inverse distance. The dielectric function can be written as follow
\begin{equation}
\varepsilon({\bf q}) = \frac{q_0^2 + {\bf q}^2}{\frac{q_0^2 \sin(|{\bf q}|R_s)}{\varepsilon(0) |{\bf q}|R_s} + {\bf q}^2}.
\label{eq:Resta-dielec}
\end{equation}
The three material parameters $q_0$, $R_s$ and $\varepsilon(0)$ in the above equation are related by Eq.~\eqref{eq:Resta-Constr} thus only two are needed for the input. The static dielectric constant $\varepsilon(0)$ and Fermi momentum related quantity $q_0$ can be extracted from the experimental data. In Resta's original paper, he offered the input parameters for Diamond, Silicon and Germanium. He further showed that the calculated dielectric functions for these materials are in close agreement with those derived from Penn-model results of Srinivasan \cite{Penn1962,Sri1969} and RPA calculations of Walter and Cohen \cite{Cohen1970}. However, the dielectric function he proposed is much simpler in the expression compared with others. Later on, Shajan and Mahadevan \cite{Mahadevan1992} used Resta's model to calculate the dielectric function of many binary semiconductors, such as GaAs, InP, ZnS, etc. Their results are found to be in excellent agreement with those calculated by the empirical pseudopotential method \cite{Richard1971}.

Here, we proposed the Resta's preconditioner by inverting Eq.~\eqref{eq:Resta-dielec}
\begin{equation}
H_{1}^{Res}({\bf q}) = \frac{\frac{q_0^2 \sin(|{\bf q}|R_s)}{\varepsilon(0) |{\bf q}|R_s } + {\bf q}^2}{q_0^2 + {\bf q}^2}.
\label{eq:Resta-Precon}
\end{equation}

It is instructive to compare this preconditioner with Kerker preconditioner. These two preconditioners are plotted as the function of $\bf q$ in Fig.~\ref{fig:Precon}. For Kerker preconditioner, the $k_{TF}$ is chosen to be 1 \AA$^{-1}$. For the Resta preconditioner, the static dielectric constant is chosen to be $6.5$. For many semiconductors and insulators, this value falls into the range of $5 \sim 15$. The $q_0$ is chosen to be 1 \AA$^{-1}$ and the screening length is 4 \AA, accordingly. These values are about the typical inputs for all binary semiconductors studied in \cite{Mahadevan1992}.
\begin{figure}[t]
 \includegraphics{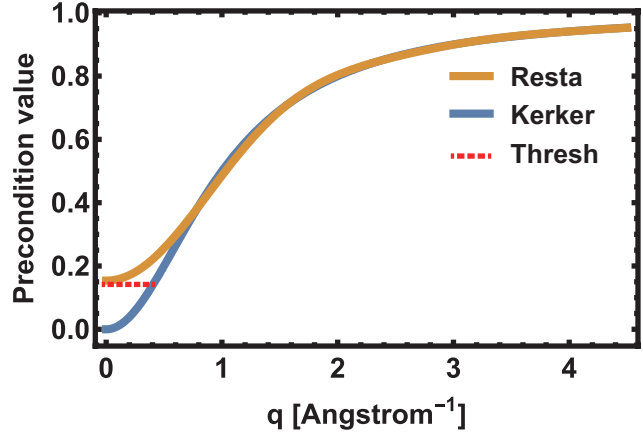}
  \caption{\label{fig:Precon} (color online) The preconditioning models as a function of reciprocal vector ${\bf q}$. A threshold parameter can be added to the Kerker preconditioner to simulate the small ${\bf q}$ behavior of the insulating systems.}
\end{figure}
From Fig.~\ref{fig:Precon}, we would like to point out the following points:\\
(i) The essential difference between the Kerker preconditioner and
Resta preconditioner lies at the long wavelength limit. Kerker preconditioner,
as we have discussed previously, goes to zero quadratically while the Resta
preconditioner goes to $1/\varepsilon(0)$. This represents the incomplete
screening in the insulating systems due to a lack of conducting electrons. If
a nominal insulating system contains the defect states which are partially filled, Resta preconditioner becomes less effective.\\
(ii) A threshold can be added to the Kerker preconditioner to mimic the behavior of the Resta preconditioner at the small ${\bf q}$'s, as shown by the dashed line in Fig.~\ref{fig:Precon}. Now the modified Kerker preconditioner takes the form:
\begin{equation}
H_{1}^{TF'}({\bf q}) = \max(a_0, \frac{{\bf q}^2}{{\bf q}^2 + k_{TF}^2}).
\label{eq:revised-Kerker}
\end{equation}
This action restores the long-range screening behavior of the insulating systems. A more practical variant includes the linear mixing parameter $\alpha$ together with the preconditioner:
\begin{equation}
H_{1}^{TF''}({\bf q}) = \max(a_0, \alpha \frac{{\bf q}^2}{{\bf q}^2 + k_{TF}^2}).
\label{eq:revised-Kerker-p}
\end{equation}
Accordingly, the optimal $a_0$ should be around $\alpha/\varepsilon(0)$. This modification extends the applicability of Kerker preconditioner to insulating systems.

\section{IV. Numerical examples}
We perform the convergence tests using the in-house
code CESSP \cite{WangHan2016,Gao2017} under the infrastructure of JASMIN
\cite{Mo2010}. The exchange and correlation energy is described by the
generalized gradient approximation proposed by Perdew, Burke, and Ernzerhof
\cite{Perdew96}. Electron-ion interactions are treated with projector
augmented wave potentials \cite{Kresse99}. The first 5 steps of the calculation take the block
variant \cite{Liu1978} of the Davison algorithm with no charge mixing. The
following steps take the RMM-DIIS method \cite{KresseFurthmuller1996} with
Pulay charge mixing. The mixing parameter $\alpha$ is set to 0.4 in all calculations with different preconditioners. The convergence criterion for self-consistent field loop is $1 \times 10^{-6} $ eV, which is sufficient for most slab calculations. The slab models include at least 20 \AA \ of vacuum layer to exclude the spurious interaction under the periodic boundary condition. Along the \emph{x} and \emph{y} directions the cells are kept as primitive cell in all calculations.

\subsection{A. Au slab: the metallic system}
The first system is \{111\} Au slab. We construct three Au slab
systems with 14, 33 and 54 layers of Au \{111\} planes, corresponding to a cell
parameter of 50, 110 and 150 {\AA} along the direction normal to Au \{111\} surface, respectively. A 12$\times$12$\times$1 k-point grid is used to sample the Brillouin zone. The cutoff energy is 350 eV. We take the modified Kerker preconditioner with $a_0=0$, referred as "original" Kerker preconditioner. The $k_{TF}$ has been set to 1  \AA$^{-1}$.

When the Kerker preconditioner is applied, the number of SCF iteration steps are 27, 32 and 31 for 14, 33 and 54 layer Au slabs, respectively. This number is weakly dependent on the size of the system, which implies that the charge sloshing has been well suppressed. As stated in the previous section, the Thomas-Fermi model and the Kerker preconditioner catch the asymptotic behavior of the dielectric function at long wavelength limit, even though the free electron gas model is not a good approximation for Au and most metallic systems. In addition, if we try Pulay mixing scheme with the preconditioning matrix $\alpha I$, the SCF convergence cannot be reached within 120 steps for any slabs.

\subsection{B. MoS$_{2}$: the layered insulating system}
Secondly, we study the convergence of the layered insulating system: MoS$_{2}$. Two slab systems with 10 and 20 layers of MoS$_{2}$, corresponding to a cell parameter of 80 {\AA}
and 160 {\AA}, have been constructed. The MoS$_{2}$ layers are stacked in the same fashion as those in the bulk MoS$_{2}$. A 6$\times$6$\times$1 k-point grid has been used to sample the Brillouin zone. The cutoff energy is 450 eV. Both the Kerker and the Resta preconditioners have been tested on the MoS$_{2}$ slab systems.

For the Resta preconditioner, we need the static dielectric constant and the screening length as input parameters. We find the reported average static elastic constant of MoS$_{2}$ depending on the number of MoS$_{2}$ layers from literature \cite{Wirtz2011,Lambrecht2012,Meunier2014,Ghosh2016}. However, they all fall into the range of 5 $\sim$ 15. In the calculations we use three static dielectric constants 5, 10 and 15 to construct the Resta preconditioner. The screening length $R_s$ has been set to 3.5 which is close to the lattice constants. The $q_0$ in the Resta model is then calculated by Eq.~\eqref{eq:Resta-Constr}.

We compare it with the original and the modified Kerker preconditioners. In these two preconditioners, the Thomas-Fermi vector $k_{TF}$ has been set to 1 \AA$^{-1}$. In the modified Kerker preconditioner, we have chosen the threshold parameters $a_0$ to be $0.4/5, 0.4/10$ and $0.4/15$ according to Eq.~\eqref{eq:revised-Kerker-p}.

\begin{table}[h]
\caption{\label{tab:MoS2}%
The number of convergence steps in the MoS$_{2}$ slab systems.
}
\begin{ruledtabular}
\begin{tabular}{ccc}
\textrm{Precondition model} & \textrm{10 layer MoS$_{2}$} & \textrm{20 layer MoS$_{2}$} \\
 \colrule
 Original Kerker & 38  & 52 \\
 Resta($\varepsilon(0)$ = 5) & 25  & 26  \\
 Resta($\varepsilon(0)$ = 10) & 30  & 31  \\
 Resta($\varepsilon(0)$ = 15) & 32  & 32  \\
 Modified Kerker ($a_0=0.4/5$) & 27 & 27 \\
 Modified Kerker ($a_0=0.4/10$) & 28 & 32 \\
 Modified Kerker ($a_0=0.4/15$) & 28 & 32 \\
\end{tabular}
\end{ruledtabular}
\end{table}

It can be seen that the Resta model and the modified Kerker model converge faster than the original Kerker scheme. This is due to a correct description of the incomplete screening effect for insulators at small ${\bf q}$. In addition, using 5, 10 or 15 for the static dielectric constant gives similar results, indicating that the convergence speed is less sensitive to this parameter.

\subsection{C. Si slab: the insulating system containing defect states}
Even though Resta preconditioner seems to be more appropriate for insulating
systems, we show that this might not be the case for the "nominal" insulating
systems containing defect states that cross the Fermi level. To illustrate this,
we construct a 96 layer Si slab with the \{111\} orientation and a cell parameter of 175 \AA \ along \emph{z} direction. Both the top and the bottom Si surfaces have one dangling bond due to the creation of the surface. A 6$\times$6$\times$1 k-point grid has been used to sample the Brillouin zone. The cutoff energy is 320 eV. The dielectric constant of bulk Si is about 12. The screening length $R_s$ is set to 4.2 \AA \ and the $q_0$ is set to 1.1 \AA$^{-1}$ according to Resta's work \cite{Resta1977}. We compare the convergence speed between three preconditioning models in Table \ref{tab:Si}.
\begin{table}[h]
\caption{\label{tab:Si}%
The number of convergence steps in the original and H-passivated Si slab systems.
}
\begin{ruledtabular}
\begin{tabular}{ccc}
\textrm{Preconditioning model} & \textrm{Bare Si slab}  &  \textrm{H-passivated Si slab} \\
 \colrule
 Original Kerker & 40 & 46 \\
 Modified Kerker ($a_0$ = 0.4/12)  & 52 & 29 \\
 Resta & 47 & 30 \\
\end{tabular}
\end{ruledtabular}
\end{table}

The original Kerker preconditioner offers the fastest convergence compared with the other two, which goes against with the conclusion from previous section. After careful inspection, we conclude that it is the surface states of the Si slab that deviate the system from a "perfect" insulating system. The density of states (DOS) of the slab and the partial charge density of the states near Fermi level have been plotted in Fig.~\ref{fig:SiComb}(a).
\begin{figure}[h]
 \includegraphics{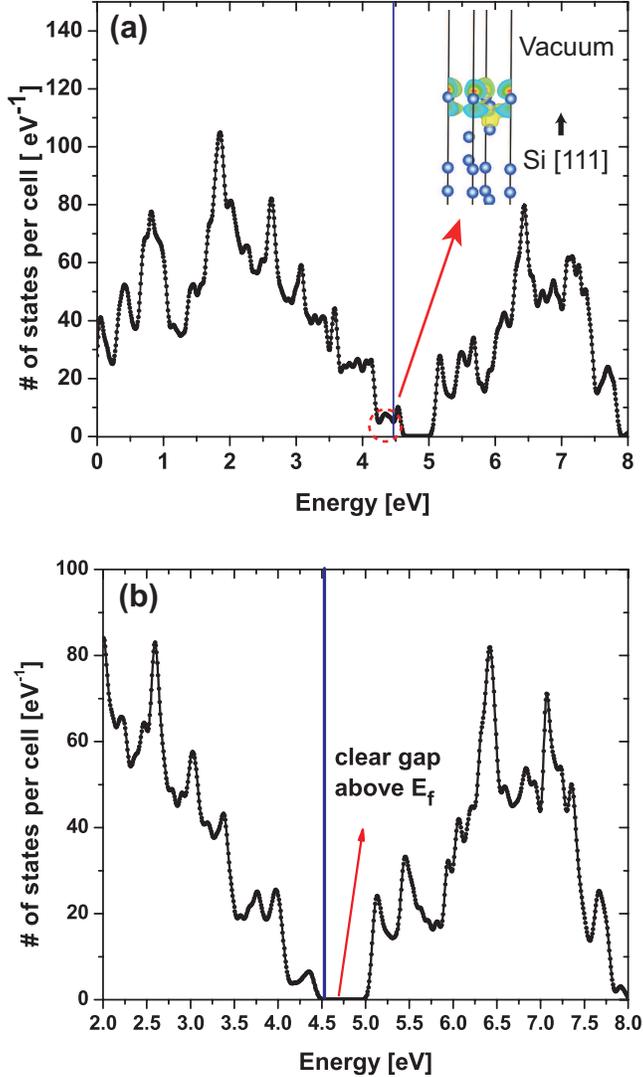}
  \caption{\label{fig:SiComb} (color online) (a) The DOS of the 96-atom Si slab. The vertical blue line indicates the Fermi level. The states right below the Fermi level are the surface states, as shown by the partial charge density plot. The bottom surface is identical to the top surface thus only one is shown. (b) The DOS of the Si slab with H passivation (96 layers of Si with 2 H passivation layers on the top and the bottom Si surfaces). The added H layers remove the surface states: a clear band gap now occurs right above Fermi level.}
\end{figure}
The creation of the surface introduces defect states right at
the Fermi level. The presence of these states drives the system away from a
"perfect" insulating system, since the number of states right at the Fermi level is
finite. This essential difference makes the preconditioning models designed
for insulators much less effective. In our previous case, on the other hand, we do not introduce surface states when creating MoS$_{2}$ slabs from the bulk due to its intrinsic layered geometry.

To further prove our idea, we passivate the Si surface states by covering the surface with H atoms. Now the system contains 96 layer of Si with 2 extra layer of H covering the top and the bottom Si surfaces. The convergence speed versus different preconditioning models is also shown in the Table \ref{tab:Si}. Now the trend is consistent with that of MoS$_{2}$: the modified Kerker (29 steps) and Resta models (30 steps) is faster than the original Kerker model (46 steps). The extra H layers have passivated the dangling bonds on the Si surfaces thus removed the surface states. This is clearly shown by the DOS of the H passivated Si slab in Fig.~\ref{fig:SiComb}(b).

Given this, the modified Kerker preconditioner and Resta preconditioner are better suited for the "prefect" insulating systems. However, introducing defect states that cross the Fermi level would render these preconditioners much less effective.

\subsection{D. Au-MoS$_{2}$: the metal-insulator hybrid system}
Now we discuss Au-MoS$_{2}$ contact systems which combine multiple layers of Au in \{111\} orientation and multiple layers of MoS$_{2}$. The Au layers and MoS$_{2}$ layers are in close contact, separated by a distance of the covalent bond length. Such structural models have been studied using DFT calculations to understand the surface, interface and contact properties of Au-MoS$_{2}$ epitaxial systems \cite{Tomanek,ZhouCSE2016,Kang2014}. The Au-MoS$_{2}$ contact configuration is similar to the \{111\} orientation configuration in Ref.~\cite{ZhouCSE2016}. To investigate the performance of the preconditioners, we have constructed slab systems that are much thicker.

Au-MoS$_{2}$ slabs with different proportion of Au and MoS$_{2}$ have been created. These slab systems share same cell parameter and nearly same slab thickness and total number of atoms. The total length of the cell is 160 \AA \ with $\sim$ 25 \AA \ vacuum layer and $\sim$ 135 \AA \ Au-MoS$_{2}$ slab. The total number of atoms is about 65 in all slabs. We use the following notation to label slabs with different proportion of Au and MoS$_{2}$: X Au + Y MoS$_{2}$ means we have X layers of Au and Y layers of MoS$_{2}$ in the slab. The total number of atoms is X + 3Y since each MoS$_{2}$ layer contains 3 layers of atoms. A 6$\times$6$\times$1 k-point grid is used to sample the Brillouin zone. The cutoff energy is 450 eV. Here we only consider the Kerker preconditioner and its modified version. Resta model is no longer appropriate to describe Au-MoS$_{2}$ hybrid systems. As shown later, it is possible to achieve fast convergence of such highly inhomogeneous systems under the Kerker preconditioning model, even though the model is originally based on the homogeneous electron gas.

There are two parameters in the modified Kerker preconditioner to adjust: $a_0$ and $k_{TF}$, according to Eqs.~\eqref{eq:original-kerker} and \eqref{eq:revised-Kerker}. Since these two parameters are describing the effectiveness of the screening from different perspective, we will be only adjusting one parameter while keeping the other fixed.

Firstly we keep $k_{TF}=1$ \AA$^{-1}$ and estimate the lower and the upper bounds of the threshold parameter $a_0$. We consider two extremes when the system is solely Au or solely MoS$_2$. In the former case, $a_0$ can be chosen as any value below $\alpha\frac{q_{min}^2}{q_{min}^2+k^2_{TF}}$, where $q_{min}$ is the smallest reciprocal vector along \emph{z} direction. Thus the lower bound is about $6 \times 10^{-4}$ by setting $q_{min}=2\pi/L$ and $L=160$ \AA. In the latter case, $a_0$ can be set as 0.04 with the static dielectric constant being set to 10. Nevertheless, it is difficult to determine the optimized value of $a_0$ for systems with varying Au proportion. Reducing $a_0$ to the original Kerker preconditioner could lead to convergence in all cases, though generally it is not the most efficient choice. Consequently the lower bound can be regarded as a safe choice.

Secondly we keep $a_0=0$ and adjust $k_{TF}$. According to Eq.~\eqref{eq:kTF_n}, $k_{TF}$ is related to the number of electrons participating in screening. Starting from the system of solely Au, increasing the proportion of MoS$_{2}$ part means a reduction in the number of "free" electrons for screening, and results in a decreasing value of $k_{TF}$. For solely Au slab, the value of $k_{TF}$ is 1 \AA$^{-1}$. Replacing a fraction, say $1-f$, of the Au slab with MoS$_{2}$, reduces the number of free electrons from $n$ to $fn$ since the MoS$_{2}$ makes no contributions to conducting electrons. Therefore, the $k_{TF}$ for the hybrid system can be estimated as $k_{TF} \times f^{1/6}$, according to Eq.~\eqref{eq:kTF_n}. For example, in the 27 Au + 12 MoS$_{2}$ system, there are total 63 atoms in the system. Thus the Au fraction is 27/63 and the corresponding $k_{TF}$ is given by $(\frac{27}{63})^{1/6} \sim 0.87$. The convergence tests results of adjusting $k_{TF}$ in the above way are listed in Table~\ref{tab:AuMoBm}, together with those from the original Kerker scheme.
\begin{table}[ht]
\caption{\label{tab:AuMoBm}%
The convergence steps in Au-MoS$_{2}$ hybrid systems.
}
\begin{ruledtabular}
\begin{tabular}{ccc}
\textrm{Au-MoS${_2}$ systems} & Original Kerker & Adjusting $k_{TF}$ \footnote{The estimated values of $k_{TF}$ are shown in the parenthesis.}  \\
 \colrule
 43 Au + 6 MoS$_{2}$ & 34 & 33 (0.94) \\
 39 Au + 8 MoS$_{2}$ & 41 & 39 (0.92) \\
 27 Au + 12 MoS$_{2}$ & 49 & 29 (0.87) \\
 16 Au + 16 MoS$_{2}$  & 43 & 41 (0.8) \\
 11 Au + 18 MoS$_{2}$  & 32 & 31 (0.74) \\
 7 Au + 19 MoS$_{2}$  & 48 & 31 (0.70) \\
 5 Au + 20 MoS$_{2}$  & 47& 34 (0.65) \\
 3 Au + 21 MoS$_{2}$  & 37 & 45 (0.60) \\
 1 Au + 22 MoS$_{2}$  & 37 & 22 (0.50) \\
\end{tabular}
\end{ruledtabular}
\end{table}

From the table, adjusting $k_{TF}$ offers at least comparable and most likely faster convergence compared with original Kerker scheme. In the low Au proportion slabs (1, 3, 5, and 7 layers), our scheme saves about 22\% of overall SCF steps compared with the original Kerker scheme. With increasing Au proportion, these two schemes exhibit similar performance as the slabs now behave more closely to bulk metals. In our opinion, adjusting $k_{TF}$ would be potentially useful in some kind of high-throughput calculations.

In the 3 Au + 21 MoS$_{2}$ system, the estimated $k_{TF}$ does not improve the convergence compared to the original Kerker scheme. Since we ignore the contribution of interface states to the "effective" free electrons, a slight increase of $k_{TF}$ could improve the preconditioner. Indeed, when changing $k_{TF}$ from 0.6 to 0.65, the convergence steps become 32, faster than the 37 steps from original Kerker scheme. Similarly, in the 16 Au + 16 MoS$_{2}$ system, changing $k_{TF}$ from 0.8 to 0.85 reduces the convergence steps from 41 to 31. We further note that applying this parameterization scheme requires \emph{a priori} knowledge of the system. The parameterization scheme for situations without sufficient \emph{a priori} knowledge will be discussed later.

\section{V. Further discussions}
We would like to address few important points and present some further discussions in this section:

\textbf{1.The key feature of the modified Kerker preconditioner}

The Thomas-Fermi screening model and the Kerker preconditioner is rooted in
the homogeneous electron gas model. It is shown by numerical examples that
with some simple but physically meaningful modifications, the Kerker
preconditioner can be applied to a wide range of materials. All test systems
are no way near the free electron gas system, such as the insulating systems
and the metal-insulator contact systems. Then what is the merit in the
modified Kerker preconditioner? We believe that a good description of the
long-range screening behavior is key to fast convergence. While in the
modified Kerker scheme, it is possible to capture the essence of long-range
screening: the original Kerker scheme naturally suppresses quadratic divergence
as ${\bf q} \rightarrow 0$ in the metallic system; the incomplete screening
effect in the insulating systems is represented by the threshold parameter
$a_0$; in the metal-insulator contact system, the long-range screening effect
is characterized by the parameter $k_{TF}$ which represents the number of
effective electrons participating in screening. The numerical examples
indeed prove the effectiveness of the modified Kerker preconditioner:
converging a large-scale slab system (with more than 60 layer and more than 150 {\AA}
long in cell parameter) to relatively high accuracy in about 30 SCF steps is significant for practical applications. Also, in many Au-MoS$_{2}$ cases, the modified Kerker scheme (when $k_{TF}$ is reasonably set) speed up 40\% compared to the original Kerker scheme.

\textbf{2. \emph{A posteriori} indicator and self-adaptive configuration}

In practice, it may be difficult to appropriately parameterize the preconditioner when lacking  \emph{a priori} knowledge. However, we can still monitor if the charge sloshing occurs during the SCF iterations by an \emph{a posteriori} indicator. Theoretically, the charge sloshing is indicated by the spectrum of the matrix $JP$ or its inverse $(JP)^{-1}$ from Eq.~\eqref{eq:error-propagation}. Practically, we could compute the eigenvalues of $P^{-1}H_m$ instead of $(JP)^{-1}$. The preconditioning matrix $P$ is a symmetric positive definite with the Kerker scheme or our modified version. The matrix $H_m$ updated by Eq.~\eqref{eq:Pulay-update2} satisfies the constraint condition in Eq.~\eqref{eq:opt-3}:
\begin{equation}
    H_mY_{m-1} = -S_{m-1}.
    \label{eq:secant}
\end{equation}
Assuming the vectors ${\bf x}_{m-i}$ all sufficiently close to the solution of Eq.~\eqref{eq:fixed-point-problem}, we have
\begin{equation}
    J^{-1}Y_{m-1} \approx -S_{m-1}.
    \label{eq:approx-secant}
\end{equation}
Comparing Eq.~\eqref{eq:secant} with Eq.~\eqref{eq:approx-secant}, we find that $H_m$ is almost the best approximation of the inverse Jacobian $J^{-1}$ in the subspace spanned by $Y_{m-1}$. Consequently the eigenvalues of $P^{-1}H_m$ are calculated in this subspace by solving the following generalized eigenvalue problem:
\begin{equation}
    Y_{m-1}^TH_mY_{m-1}{\bf u}_i = \lambda_i Y_{m-1}^TPY_{m-1}{\bf u}_i.
    \label{eq:eigenpro}
\end{equation}
In implementation, we shift Eq.~\eqref{eq:eigenpro} as like
\begin{equation}
    Y_{m-1}^T(H_m-P)Y_{m-1}{\bf u}_i = (\lambda_i-1) Y_{m-1}^TPY_{m-1}{\bf u}_i.
    \label{eq:eigenpro-shift}
\end{equation}
Note that it takes little computational overhead to solve Eq.~\eqref{eq:eigenpro-shift} since $(H_m-P)Y_{m-1}=-(S_{m-1}+H_1Y_{m-1})$ has been calculated in Pulay's update and the dimension of Eq.~\eqref{eq:eigenpro-shift} is generally less than 50 in our code.

As far as we know, Kresse and Furthm{\"u}ller \cite{KresseFurthmuller1996} propose similar formula as Eqs.~\eqref{eq:eigenpro} and \eqref{eq:eigenpro-shift} to investigate the spectrum range for insulators and open-shell transition metals of different sizes. Instead of examining the range of spectrum, we extract the minimal module of the eigenvalues from Eq.~\eqref{eq:eigenpro}. In principle, charge sloshing directly causes a divergence trend in the eigenvalues of the dielectric matrix $J$. If the charge sloshing is not suppressed by the preconditioner $P$, it will give rise to some large eigenvalues in the spectrum of $JP$.  Then the least modulus of the spectrum of $(JP)^{-1}$ would be small. As discussed above, we approximate $J^{-1}$ by $H_m$ in the subspace spanned by $Y_{m-1}$. So the least modulus of the eigenvalues of Eq.~\eqref{eq:eigenpro} or \eqref{eq:eigenpro-shift} can be chosen as the \emph{a posteriori} indicator to show whether the preconditioner has suppressed charge sloshing or not. We believe this quantity is more directly related to the occurrence of charge sloshing than the range of spectrum. Based on our experience, if charge sloshing occurs in the practical calculation, the \emph{a posteriori} indicator would be generally below 0.1. With the \emph{a posteriori} indicator, we further realize the self-adaptive configuration of the SCF iteration.

We demonstrate the practical use of the \emph{a posteriori} indicator in calculations of the 5 Au + 20 MoS$_{2}$ system. The energy convergence and the \emph{a posteriori} indicator during the SCF calculation are plotted in Fig.~\ref{fig:5Au}.
\begin{figure}[ht]
 \includegraphics{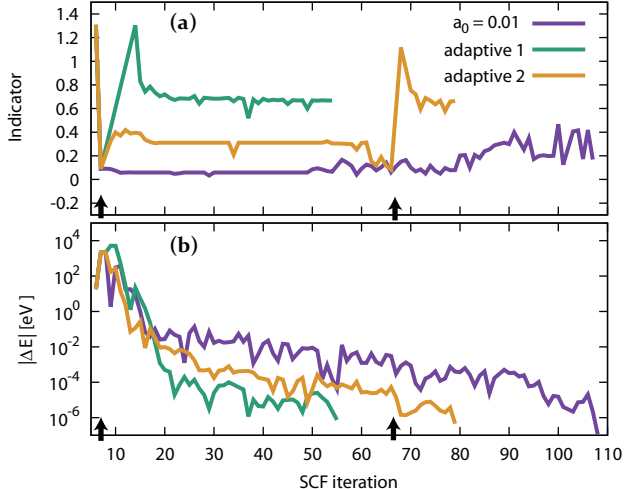}
  \caption{\label{fig:5Au} (color online) The \emph{a posteriori} indicator (a) and the energy convergence (b) versus SCF iteration step under $a_0 = 0.01$,
  and two self-adaptive schemes are shown. The arrows on the axis
  indicate that we launch self-adaptive configuration at the 8th step and 67th step.}
\end{figure}
Before doing the calculation, one would guess the 5 Au + 20 MoS$_{2}$ system is similar
to the solely MoS$_{2}$ system since the major part is MoS$_{2}$. Then, we
begin with $a_0=0.01$ (note $a_0$ is the threshold parameter in the preconditioner, which should be distinguished from the indicator). It is shown in
Fig.~\ref{fig:5Au} that the SCF convergence with $a_0=0.01$ is slow. Meanwhile, the
\emph{a posteriori} indicator is lying below 0.1 at most of the first 80 SCF
steps, which implies an incomplete suppression of the charge sloshing.

We design two self-adaptive schemes when the \emph{a posteriori} indicator falls below threshold 0.1. One is to stop the current task and restart the calculation with original Kerker preconditioner (corresponding to
"adaptive 1" in Fig.~\ref{fig:5Au}). After the self-adaptive configuration
at the 8th step, the \emph{a posteriori} indicator is kept around 0.7 and it saves about a half of the SCF iteration steps compared with the $a_0=0.01$ run. The other way is to clear the subspace of $Y_{m-1}$ (the information from previous steps) and continue the SCF iteration with original Kerker preconditioner (corresponding to "adaptive 2" in Fig.~\ref{fig:5Au}). In this case, two reconfigurations occur at the 8th step
and 67th step to keep the \emph{a posteriori} indicator above 0.1. The SCF convergence is finally reached around 80 steps, still saves about 30 steps compared with the $a_0=0.01$ run. The former scheme seems more efficient than the latter for now. Further studies on the self-adaptive configuration in the SCF calculations will be presented in our follow-up research.

\textbf{3. Integrated preconditioning scheme}

Here we present the complete strategy of the modified Kerker preconditioning in Table \ref{tab:strategy}.
\begin{table}[ht]
\caption{\label{tab:strategy}%
The general parameterization strategy for large dimension systems.
}
\begin{ruledtabular}
\begin{tabular}{cccc}
 No. & System & Long-range screening properties & Preconditioner \\
  \colrule
 1 & Metal & $1/q^{2}$ & Original Kerker \\
 2 & Insulator & $\varepsilon(0)$ & Modified Kerker or Resta \\
 3 & Metal + insulator & Effective "free" $e^{-}$ & $k_{TF}=f^{1/6}$ \\
 4 & Unknown & Unknown & $a_0 \sim 0.04$ \& \emph{a posteriori} indicator \\
\end{tabular}
\end{ruledtabular}
\end{table}
We add some remarks on this integrated strategy: \\
(i) The threshold parameter $a_0$ for insulators can be set to 0.04 as default. The static dielectric constant for most insulators falls into the range between 5 $\sim$ 15, and the SCF convergence is not that
sensitive to static dielectric constant. Therefore, we expect that the default setting could help to achieve fast convergence in many insulating systems. \\
(ii) We discuss the metal-insulator contact systems where the metal region and insulator
region are spatially separated and well defined. But a fine mixing of them on the scale of atomic distance does not fall into this category. Such situation should be treated as
a system lack of \emph{a priori} knowledge unless further information can be founded.\\
(iii) Strategy 4 basically presumes the system is insulator. Then the SCF iteration is monitored by the \emph{a posteriori} indicator. If the charge sloshing occurs, the preconditioning scheme could be self-adaptively reconfigured. For now we suggest using "adaptive 1", which discards the current calculation and restart with original Kerker preconditioner. However, we expect to develop more efficient self-adaptive schemes in the future studies.\\

\section{VI. Conclusions}
We have proposed the modified Kerker scheme to improve the SCF convergence for metallic,
insulating and metal-insulator hybrid systems. The modifications contain
following key points: the original Kerker preconditioner is suited for typical metallic systems; the threshold parameter $a_0$ characterizes the screening
behavior of insulators at long wavelength limit thus helps to accommodate the
insulating systems; the $k_{TF}$ represents the effective number of conducting
electrons and its approximation can be used to improve
the SCF convergence for metal-insulator hybrid systems; the \emph{a
posteriori} indicator guides the inexperienced users away from staggering into
the charge sloshing. These modifications cost negligible extra computation overhead and exhibit the flexibility of working in either \emph{a priori} or self-adaptive way, which would be favored
by the high-throughput first-principles calculations.

\section{Acknowledgements}
This work was partially supported by Science Challenge Project under Grant
JCKY2016212A502, the National Key Research and Development Program of China under Grant
2016YFB0201204, the National Science Foundation of China under Grants
91730302 and 11501039, the China Postdoctoral Science Foundation under Grant 2017M610820.

\section{Appendix}
Now we prove that Eq.~\eqref{eq:Broyden-update} is the solution to the constraint
optimization problem Eq.~\eqref{eq:opt-3}. It is prerequisite to prove the following lemma.
\begin{lemma}
    \label{lemma:mini-norm}
    Let $X\in\mathbb{C}^{m\times n}, A\in\mathbb{C}^{n\times p},
    B\in\mathbb{C}^{m\times p}$, and assume that $A$ has full column rank.
    Denote the Moore-Penrose pseudoinverse of $A$ by
    $A^{\dagger}$ with $A=\left(A^HA\right)^{-1}A^H$. If $XA=B$ is satisfiable, and the
    matrix $Z\equiv BA^{\dagger}$, then it holds that
    \begin{equation}
        \|Z\|_F \leqslant \|X\|_F
        \label{eq:mini-norm}
    \end{equation}
\end{lemma}
\begin{proof}
    Let $Q\equiv AA^{\dagger}$. Then it follows that $Q^H=Q$ and $ZQ=Z$.
    Thus we have
    \begin{eqnarray}
        \left(X-Z, Z\right)_F &\equiv& \mathrm{tr}\left[(X-Z)Z^H\right] \\
        &=& \mathrm{tr}\left[(X-Z)(ZQ)^H\right] \nonumber \\
        &=& \mathrm{tr}\left[(X-Z)Q^HZ^H\right] \nonumber \\
        &=& \mathrm{tr}\left[(X-Z)QZ^H\right] \nonumber \\
        &=& \mathrm{tr}\left[\left(XAA^{\dagger}-ZQ\right)Z^H\right] \nonumber \\
        &=& \mathrm{tr}\left[\left(BA^{\dagger}-Z\right)Z^H\right] \nonumber \\
        &=& 0 \nonumber
        \label{eq:Frobenius-inner-product}
    \end{eqnarray}
    Note that Eq.~\eqref{eq:Frobenius-inner-product} is an inner product
    corresponding to the Frobenius norm $\|\cdot\|_F$. Hence
    \begin{eqnarray*}
        \|X\|_F^2 &=& \|X-Z\|_F^2 + 2\left(X-Z, Z\right)_F + \|Z\|_F^2 \\
        &=& \|X-Z\|_F^2 + \|Z\|_F^2 \\
        &\geqslant& \|Z\|_F^2
        \label{}
    \end{eqnarray*}
    with equality if and only if $X=Z$.
\end{proof}

\vspace{5pt}
Let $H'\equiv H-H_{m-1}$. Thus the optimization problem \eqref{eq:opt-3} can
be replaced by its equivalent one
\begin{equation}
    \left\{
    \begin{array}{rl}
        \min_{H'} & \|H'\|_F^2 \\
        \textrm{s.t.} & H'Y_{m-1} = -\left(S_{m-1}+H_{m-1}Y_{m-1}\right),
    \end{array}
    \right.%
    \label{eq:opt-4}
\end{equation}
It follows from Lemma \ref{lemma:mini-norm} that the solution to the problem
Eq.~\eqref{eq:opt-4} is
\begin{equation}
    H'= -\left(S_{m-1}+H_{m-1}Y_{m-1}\right)\left(Y_{m-1}^TY_{m-1}\right)^{-1}Y_{m-1}^T.
    \label{eq:opt-4-solu}
\end{equation}
Therefore the solution to the problem \eqref{eq:opt-3} is
\begin{equation}
    H = H_{m-1} -\left(S_{m-1}+H_{m-1}Y_{m-1}\right)\left(Y_{m-1}^TY_{m-1}\right)^{-1}Y_{m-1}^T.
    \label{eq:opt-3-solu}
\end{equation}


\bibliography{Zhouetal}

\end{document}